\documentclass[11pt]{article}
\usepackage{amsfonts,amsthm, amsmath}
\usepackage{epsfig}
\usepackage{makeidx}
\usepackage{graphicx}
\usepackage{graphicx,epstopdf}
\DeclareGraphicsExtensions{.ps,.png,.jpg}

\makeindex
\usepackage{color}

\newcommand{\ncom}{\newcommand}
\ncom{\ul}{\underline}
\ncom{\beq}{\begin{equation}}
\ncom{\eeq}{\end{equation}}
\ncom{\bea}{\begin{eqnarray*}}
\ncom{\eea}{\end{eqnarray*}}
\ncom{\beqa}{\begin{eqnarray}}
\ncom{\eeqa}{\end{eqnarray}}
\ncom{\nno}{\nonumber}
\ncom{\non}{\nonumber}
\ncom{\ds}{\displaystyle}
\ncom{\half}{\frac{1}{2}}
\ncom{\mbx}{\makebox{.25cm}}
\ncom{\hs}{\mbox{\hspace{.25cm}}}
\ncom{\rar}{\rightarrow}
\ncom{\Rar}{\Rightarrow}
\ncom{\noin}{\noindent}
\ncom{\bc}{\begin{center}}
\ncom{\ec}{\end{center}}
\ncom{\sz}{\scriptsize}
\ncom{\rf}{\ref}
\ncom{\s}{\sqrt{2}}
\ncom{\sgm}{\sigma}
\ncom{\Sgm}{\Sigma}
\ncom{\psgm}{\sigma^{\prime}}
\ncom{\dt}{\delta}
\ncom{\Dt}{\Delta}
\ncom{\lmd}{\lambda}
\ncom{\Lmd}{\Lambda}
\ncom{\Th}{\Theta}
\ncom{\e}{\eta}
\ncom{\eps}{\epsilon}
\ncom{\pcc}{\stackrel{P}{>}}
\ncom{\lp}{\stackrel{L_{p}}{>}}
\ncom{\dist}{{\rm\,dist}}
\ncom{\sspan}{{\rm\,span}}
\ncom{\re}{{\rm Re\,}}
\ncom{\im}{{\rm Im\,}}
\ncom{\sgn}{{\rm sgn\,}}
\ncom{\ba}{\begin{array}}
\ncom{\ea}{\end{array}}
\ncom{\hone}{\mbox{\hspace{1em}}}
\ncom{\htwo}{\mbox{\hspace{2em}}}
\ncom{\hthree}{\mbox{\hspace{3em}}}
\ncom{\hfour}{\mbox{\hspace{4em}}}
\ncom{\vone}{\vskip 2ex}
\ncom{\vtwo}{\vskip 4ex}
\ncom{\vonee}{\vskip 1.5ex}
\ncom{\vthree}{\vskip 6ex}
\ncom{\vfour}{\vspace*{8ex}}
\ncom{\norm}{\|\;\;\|}
\ncom{\integ}[4]{\int_{#1}^{#2}\,{#3}\,d{#4}}
\ncom{\vspan}[1]{{{\rm\,span}\{ #1 \}}}
\ncom{\dm}[1]{ {\displaystyle{#1} } }
\ncom{\ri}[1]{{#1} \index{#1}}

\newtheorem{theorem}{\bf Theorem}[section]

\newtheorem{proposition}{Proposition}[section]

\newtheoremstyle
    {remarkstyle}
    {}
    {11pt}
    {}
    {}
    {\bfseries}
    {:}
    {     }
    {\thmname{#1} \thmnumber{#2} }

\theoremstyle{remarkstyle}

\begin{document}

\newpage

\begin{center}
{\Large \bf Probabilistic properties of detrended fluctuation analysis for Gaussian processes}
\end{center}
\vone
\begin{center}
 {G. Sikora}$^{\textrm{a}}$,  {M. Hoell}$^{\textrm{b}}$, {A. Wy{\l}oma{\'n}ska}$^{\textrm{a}}$, {J. Gajda}$^{\textrm{a}}$, {A.V. Chechkin}$^{\textrm{c,d}}$ and {H. Kantz}$^{\textrm{b}}$ 
{\footnotesize{
		$$\begin{tabular}{l}
		\\
		$^{\textrm{a}}$ \emph{Faculty of Pure and Applied Mathematics, Hugo Steinhaus Center, Wroc{\l}aw University of Science and  Technology,}\\\emph{Wroc{\l}aw, Poland}\\
				$^{\textrm{b}}$ \emph{Max Planck Institute for the Physics of Complex Systems, Dresden, Germany}\\
				$^{\textrm{c}}$ \emph{Institute for Physics \& Astronomy, University of Potsdam, Potsdam-Golm, Germany}\\
			$^{\textrm{d}}$ \emph{Akhiezer Institute for Theoretical Physics NSC "Kharkov Institute of Physics and Technology",}\\
						\emph{Kharkov, Ukraine}\\
		\end{tabular}$$} }
\end{center}

\begin{center}
\noindent{\bf Abstract}\\
\end{center}
The detrended fluctuation analysis (DFA) is one of the most widely used tools for the detection of long-range correlations in time series. Although DFA has found many interesting applications and has been shown as one of the best performing detrending methods, its probabilistic foundations are still unclear. In this paper we study probabilistic properties of DFA for Gaussian processes. The main attention is paid to the distribution of the squared error sum of the detrended process. This allows us to find the expected value and the variance of the fluctuation function of DFA  for a Gaussian process of  general form. The  results obtained can serve as a starting point for analyzing the statistical properties of the DFA-based estimators for the fluctuation and correlation parameters. The obtained theoretical formulas are supported by numerical simulations of particular Gaussian processes possessing short-and long-memory behaviour.

\section{Introduction}

The detrended fluctuation analysis (DFA) was introduced in 1994 by Peng \textit{et al.} for analyzing DNA sequences \cite{peng}. This method appears to be efficient in eliminating  deterministic trends and is well-performing \cite{dfabashan,dfachen2,dfachen,dfahu,dfama,dfaxu}. It allows to estimate the fluctuation function which is the basic quantity of DFA.  In most of the research papers the DFA serves as a quantifier for the classification of the  fluctuation parameter $\alpha$. The DFA statistic is very popular in various fields of science and engineering since it is one of the most widely used methods for detection of so-called long-range correlations in time series. In general, long-range correlations are defined via the autocorrelation function $C(s)$ with time lag $s$ of a time series. If the summation over all time lags diverges then the time series is called long-range correlated. This divergence reflects the long-term memory of such processes and is often caused by a decreasing power law $C(s) \sim s^{-\delta}$ with correlation parameter $\delta$. 

If the summation converges the process is called short-range correlated. In order to estimate $\delta$ of a time series, one can use the estimator of the autocorrelation function.  Unfortunately, this approach has several practical problems especially for a small number of data points.  A possible solution is to  introduce another function which may characterize the memory properties.  One of the statistics which is useful here is the fluctuation function $F(s)$ of the DFA. This statistics provides an indirect way of estimating the correlation exponent $\delta$. It follows an increasing power law $F^2(s) \sim s^{2\alpha}$ with fluctuation parameter $\alpha$ which is connected to $\delta$. For large $s$ the power law behavior of the fluctuation function of the DFA can be seen more easily in the log-log plot than the power law of $C(s)$, because $F(s)$ is an increasing function with respect to $s$. 

Despite its success, there are only a few articles investigating the fluctuation function of DFA analytically. The fluctuation parameter $\alpha$ has been derived for fractional Gaussian noise \cite{dfabardet,dfamovahed,dfataqqu} and for fractional Brownian motion \cite{dfaheneghan,dfakiyono1,dfasunspot}. The connection between the DFA and the power spectral analysis has been investigated in \cite{dfaheneghan,dfakiyono1,dfakiyono2,dfakiyono3,dfatalkner,dfawillson1,dfawillson2}. Also, the relationship
between the fluctuation function of DFA and the autocorrelation function for stationary
processes is known  \cite{Marc1,Marc2,Marc3}. The fluctuation function of DFA and its probabilistic properties like asymptotic behavior and linear regression estimator, in the case of fractional Gaussian noise, have been studied in \cite{dfabardet} and \cite{Cra10}, respectively.  Nevertheless, there are still open questions about the basic principles and interpretation of DFA which cannot be answered by the current analytical knowledge.

In this paper we study probabilistic properties of the DFA for a general family of Gaussian processes.  Using the theory of quadratic forms of Gaussian processes we  find the distribution of the squared error sum of the detrended process. 
The main result of the paper is the Theorem \ref{th2}. In contrast to the papers mentioned above, where only the asymptotic behavior of the fluctuation function of the DFA is considered, we here present the exact formulas for the expected value and the variance of the fluctuation function for Gaussian processes. 

This paper naturally continues recent research on the time-averaged statistics for the Gaussian processes \cite{gajda} and serves as a starting point for the analysis of the statistical properties of the DFA-based estimators for fluctuation parameter $\alpha$ and  the correlation parameter $\delta$. In this article we consider the simplest case of detrending, namely detrending with order $q=1$. The presented methodology can be easily extended to any order of the detrending function.

 \section{Probabilistic properties of DFA statistics for Gaussian processes}\label{main}

In what follows we consider centered Gaussian processes. For given trajectory $\{X(1),X(2),\ldots,X(N)\}$ 
 with the covariance matrix $\Sigma=\{E[X(i)X(j)]:i,j=1,2,\ldots,N\}$ the procedure of the detrended fluctuation analysis consists of several steps. First, the time axis $1,2,\cdots, N$ is divided into $K$ segments of length $s$, $K=[N/s]$. In every segment $v$, $v=1,2,\dots,K$, we derive the variance $f^2(v,s)$ given by the squared error sum of the detrended process
\begin{align}\label{eqn:jklg}
f^2(v,s)=\frac{1}{s}\sum_{t=1+d_v}^{s+d_v}{\left[X(t)-p_v(t)\right]^2}=\frac{1}{s}\sum_{t=1}^{s}{\left[X(t+d_v)-p_v(t+d_v)\right]^2},
\end{align}
where $p_v(\cdot)$ is the fitting polynomial of order $q$ of $X(t)$ in the segment $v$ obtained by ordinary least squares method, and $d_v=(v-1)s$. The  order $q$ of $p_v(\cdot)$ is a free parameter.  The first three orders of constant, linear and quadratic detrending, $q=0$, $1$ and $2$, are the ones which are mostly used in practical applications. Finally, the square of the fluctuation function of DFA is the average over all the squared error sums
\begin{align}\label{dfa}
F^2(s) = \frac{1}{K} \sum_{v=1}^K f^2(v,s)=\frac{1}{[N/s]}\sum_{v=1}^{[N/s]}f^2(v,s) .
\end{align}
In this paper we analyze the case $q=1$, so $p_v(\cdot)=\hat{a}+\hat{b}t$ is a linear function of $t$. The coefficients of the linear fit can be calculated directly from the linear system of equations
\begin{equation}\label{first}
\begin{pmatrix}
\hat{a}_v \\ \hat{b}_v
\end{pmatrix}
=\begin{pmatrix}
S_{0,v} & S_{1,v} \\ S_{1,v} & S_{2,v}
\end{pmatrix}^{-1}
\begin{pmatrix}
\sum_{i=1+d_v}^{s+d_v} X(i)\\
\sum_{i=1+d_v}^{s+d_v} iX(i)
\end{pmatrix}
\end{equation}
with $S_{j,v}=\sum_{i=1+d_v}^{s+d_v}i^j$, $j=0,1,2$. Using equation (\ref{first}) the linear fit can be written as a weighted sum of $X(i)$, $i=1+d_v,\ldots,s+d_v$,
\begin{eqnarray}\label{2_4}
p_v(t)=\hat{a}_v+\hat{b}_vt=\sum_{i=1+d_v}^{s+d_v}X(i)\frac{S_{2,v}-iS_{i,v}+t(iS_{0,v}-S_{1,v})}{S_{0,v}S_{2,v}-S_{1,v}^2}.
\end{eqnarray} 
Let us define the weights
\begin{eqnarray}\label{we}
P_v(i,t)=\frac{S_{2,v}-iS_{i,v}+t(iS_{0,v}-S_{1,v})}{S_{0,v}S_{2,v}-S_{1,v}^2}.
\end{eqnarray}
From equations (\ref{2_4}) and (\ref{we}) we get
\begin{eqnarray}
p_v(t)=\sum_{i=1+d_v}^{s+d_v}X(i)P_v(i,t)
\end{eqnarray}
and thus
\begin{eqnarray}
X(t+d_v)-p_v(t+d_v)=\sum_{i=1}^sX(i+d_v)[\delta_{i,t}-P_v(i+d_v,t+d_v)]
\end{eqnarray}
with  $\delta_{\cdot,\cdot}$ being the Kronecker delta and $i,t=1,2,\ldots,s$. The weights (\ref{we}) can be expressed in the explicit form
\begin{eqnarray}\label{we2}
P_v(i+d_v,t+d_v)=\frac{6i(2t-s-1)+2(s+1)(-3t+2s+1)}{s^3-s}
\end{eqnarray}
which indicates that $P_v(i+d_v,t+d_v)$ does not depend on $v$. Therefore, taking the notation
\begin{eqnarray}
Y_v(t)=X(t+d_v)-p_v(t+d_v),~v=1,\ldots [N/s],~t=1,\ldots,s
\end{eqnarray}
one can write
\begin{eqnarray}\label{Y_v}
Y_v(t)=\sum_{i=1}^sX(i+d_v)[\delta_{i,t}-P_1(i,t)].
\end{eqnarray}

\begin{theorem}\label{th1}
For each $v,u=1,2,\ldots,[N/s]$ the covariance $E[Y_v(m)Y_u(n)]$, $m,n=1,2,\ldots,s$ of the vectors $\mathbb{Y}_v=\{Y_v(t):t=1,\ldots,s\}$ and $\mathbb{Y}_u=\{Y_u(t):t=1,\ldots,s\}$ has the form
\begin{eqnarray}\label{wz5}
E[Y_v(m)Y_u(n)]=\sum_{i,j=1}^sE[X(i+d_v)X(j+d_u)]\left(\delta_{i,m}-P_1(i,m)\right)\left(\delta_{j,n}-P_1(j,n)\right),
\end{eqnarray}
where $P_1(\cdot,\cdot)$ is given in (\ref{we2}) and $\delta_{\cdot,\cdot}$ is the Kronecker delta. 
\end{theorem}
\begin{proof}
Formula (\ref{wz5}) follows directly from equation (\ref{Y_v}).
\end{proof}
\noindent Let us notice that under the assumption $\{X(1),X(2),\ldots,X(N)\}$ is a centered Gaussian process, for each $v=1,2,\ldots,[N/s]$ the random variable $sf^2(v,s)$ can be represented as a quadratic form of the vector $\mathbb{Y}_v=\{Y_v(t):t=1,\ldots,s\}$, where $Y_v(t)$ is defined in (\ref{Y_v}). More precisely
\begin{equation}\nonumber
sf^2(v,s)=\mathbb{Y}_v\mathbb{Y}^T_v=\sum_{t=1}^{s}\left[X(t+d_v)-p_v(t+d_v)\right]^2.
\end{equation}
According to the Gaussian quadratic forms theory \cite{MatPro92}, $sf^2(v,s)$ has so-called generalized chi-squared distribution, namely
\begin{equation}\label{rep}
sf^2(v,s)=\mathbb{Y}_v\mathbb{Y}^T_v\stackrel{d}{=}\sum_{j=1}^s\lambda_j(v) U_j,
\end{equation}
where $U_j'$s are independent identically $\chi^2$ distributed with 1 degree of freedom random variables, and weights $\lambda_j(v)$ are the eigenvalues of the covariance matrix $\Sigma_{\mathbb{Y}_v}=\{E[Y_v(m)Y_v(n)]:~~m,n=1,\ldots,s\}$. Therefore, for random quantity $sf^2(v,s)$ the following propositions hold:

\begin{proposition}
For the centered Gaussian process $\{X(1),X(2),\ldots,X(N)\}$ and each $v=1,2,\ldots,[N/s]$, the random variable $f^2(v,s)$  has the following
\begin{itemize}
\item[a)]{expected value:
\begin{equation}\label{proa}
E\left[f^2(v,s)\right]=\frac{1}{s}tr\left(\Sigma_{\mathbb{Y}_v}\right)=\frac{1}{s}\sum_{j=1}^sE[Y_v^2(j)],
\end{equation}}
and \item[b)]{variance:
$$Var\left[f^2(v,s)\right]=\frac{2}{s^2}tr\left(\Sigma^2_{\mathbb{Y}_v}\right)=\frac{2}{s^2}\sum_{i,j=1}^s\left(E[Y_v(i)Y_v(j)]\right)^2,$$}
\end{itemize}
\end{proposition}
\begin{proof}
\noindent
\begin{itemize}
\item[a)]Under the assumption that $\{X(1),X(2),\ldots,X(N)\}$ is a centered Gaussian process the distribution of $sf^2(v,s)$ is given in \eqref{rep}. Taking under consideration the fact that for $U$ - $\chi^2$ distributed random variable with 1 degree of freedom $E[U]=1$ we get
$$E\left[f^2(v,s)\right]=\frac{1}{s}\sum_{j=1}^s\lambda_j(v).$$
Let us remind an important fact from linear algebra, namely the sum of all eigenvalues of a given  matrix is identified as the trace of that matrix. Since $\lambda_j(v)$, $j=1,\ldots,s$ are the eigenvalues of the covariance matrix $\Sigma_{\mathbb{Y}_v}$, by using the mentioned fact from linear algebra we obtain
$$E\left[f^2(v,s)\right]=\frac{1}{s}\sum_{j=1}^s\lambda_j(v)=\frac{1}{s}\sum_{j=1}^sE[Y_v^2(j)].$$
\item[b)]{First let us remind that  for $U$ - $\chi^2$ distributed random variable with 1 degree of freedom $Var[U]=2$. In order to calculate the variance of the random variable $f^2(v,s)$ we use
 two well-known facts from linear algebra. The first fact was mentioned above, namely
 the sum of all eigenvalues of a given matrix is identified with the trace of that matrix.
The second fact is that the squared eigenvalues of given matrix are the eigenvalues of this matrix taken to the power two. Under the assumption that $\{X(1),X(2),\ldots,X(N)\}$ is a general centered Gaussian process the distribution of $sf^2(v,s)$ is represented in \eqref{rep}. Using the fact that  matrix $\Sigma_{\mathbb{Y}_v}$ is symmetric and $U_j$, $j=1,\ldots,s$ in \eqref{rep} are independent, one obtains\newpage
$$Var\left[f^2(v,s)\right]=\frac{2}{s^2}\sum_{j=1}^s\lambda_j^2(v)=\frac{2}{s^2}tr\left(\Sigma^2_{\mathbb{Y}_v}\right)=\frac{2}{s^2}\sum_{i,j=1}^s\left(E[Y_v(i)Y_v(j)]\right)^2.$$
}
\end{itemize}
\end{proof}
\noindent In the next proposition, we present the main characteristics for the distribution of the random variable $f^2(v,s)$.

\begin{proposition}
For the centered Gaussian process $\{X(1),X(2),\ldots,X(N)\}$ and each $v=1,2,\ldots,[N/s]$, the random variable $f^2(v,s)$  has the following
\begin{itemize}
\item[a)]{characteristic function:
$$\phi_{v,s}(x)=E[\exp\left(if^2(v,s)x\right)]=\prod_{j=1}^s\frac{1}{\left[1-2\lambda_j(v)ixs\right]^{1/2}},$$}
\item[b)]{moment generating function:
$$MGF_{v,s}(x)=E[\exp\left(if^2(v,s)x\right)]=C \bigl(1-2\lambda_1(v)x \bigr)^{-s/2} \exp\left(\sum_{k=1}^{\infty} \frac{\gamma_k}{\bigl(1-2\lambda_1(v)x \bigr)^k}\right),$$
where
$\lambda_1(v)$ is the smallest eigenvalue of the matrix
$\Sigma_{\mathbb{Y}_v}$,
\begin{equation*}
\gamma_k=\sum_{j=1}^{s}\frac{(1-\lambda_1(v)/\lambda_{j}(v))^k}{2k},
\end{equation*}
and
\begin{equation*}
C=\prod_{j=1}^{s}\left(\frac{\lambda_1(v)}{\lambda_j(v)}\right)^{1/2},
\end{equation*}}
and \item[c)]{probability density function:
\begin{equation}\label{pdf}
g_{v,s}(x)=\!C\!\sum_{k=0}^{\infty}\frac{\Delta_k \, x^{\frac{s}{2}+k-1} \, \exp\left(-\frac{xs}{2\lambda_1(v)}\right)}
{\Gamma\left(\frac{s}{2}+k\right) \, \left(\frac{2\lambda_1(v)}{s}\right)^{\frac{s}{2}+k}},
\end{equation}
where $\Delta_k$ is expressed by the recursive formula
\begin{equation*}
\Delta_{k+1}=\frac{1}{k+1}\sum_{j=1}^{k+1} j  \gamma_j \Delta_{k+1-j}, \quad \Delta_0=1.
\end{equation*}}
\end{itemize}
\end{proposition}
\begin{proof}
\noindent
Let us notice that the distribution of the quadratic form $sf^2(v,s)$ given in \eqref{rep} can be represented as a sum of $s$ independent Gamma distributed random variables with constant shape parameter $1/2$ and different scale parameters. Namely $\lambda_j(v)U_j\stackrel{d}{=}G(1/2,2\lambda_j(v))$, \cite{Mos85}, where $G(k,\theta)$ is the Gamma distributed random variable with parameters $k$ and $\theta$. The  characteristic function of $G(k,\theta)$ random variable is given by
$$\phi_{(k,\theta)}(t)=\frac{1}{\left[1-\theta it\right]^{k}}.$$
We also remind that the probability density function  of the Gamma distributed random variable
$G(k,\theta)$ with shape parameter $k$ and scale parameter $\theta$ reads
\begin{displaymath}
h_{(k,\theta)}(x)=\frac{x^{k-1} \,  \exp(-x/\theta)}{\Gamma(k) \,\theta^k}, \quad x>0.
\end{displaymath}
\begin{itemize}
\item[a)]
Using the above facts the characteristic function of $sf^2(v,s)$ is a product of characteristic functions of Gamma distributed random variables.  
\item[b),c)]{The expressions for the moment generating function and the probability density function stem from the result of \cite{Mos85}, which establishes such quantities for a linear combination of independent gamma random variables. In our case, such linear combination is the statistic
$$f^2(v,s)=\frac{1}{s}\sum_{j=1}^s\lambda_j(v)U_j,$$
where $\lambda_j(v)U_j\stackrel{d}{=}G(1/2,2\lambda_j(v)).$}
\end{itemize}
\end{proof}
\noindent In the following theorem we present the formulas for the expected value and the variance of the DFA for a
general centered Gaussian process. 

\begin{theorem}
\label{th2}
The square $F^2(s)$ of the fluctuation function for general centered Gaussian process $\{X(1),X(2),\ldots,X(N)\}$ has the following
\begin{itemize}
\item[a)]{expected value:
\begin{eqnarray}E\left[F^2(s)\right]=\frac{1}{s[N/s]}\sum_{v=1}^{[N/s]}\sum_{j=1}^s{E\left[Y^2_v(j)\right]}\end{eqnarray}}
and \item[b)]{variance:
\begin{eqnarray}
Var\left[F^2(s)\right]=\frac{2}{s^2[N/s]^2}\sum_{1\leq v,u\leq[N/s]}\sum_{1\leq t,w\leq s}\Bigg[E\left [Y_v(t)Y_u(w)\right]\Bigg]^2,
\end{eqnarray}
where $Y_v(t)$ is defined in (\ref{Y_v}) for any $v=1,\ldots,[N/s]$ and $t=1,\ldots,s$.
}
\end{itemize}
\end{theorem}
\begin{proof}
\noindent
\begin{itemize}
\item[a)]{According to (\ref{dfa}) the expected value of $F^2(s)$ takes the form
$$E\left[F^2(s)\right]=\frac{1}{[N/s]}E\left[f^2(v,s)\right]=\frac{1}{s[N/s]}\sum_{v=1}^{[N/s]}\sum_{j=1}^s{E\left[Y^2_v(j)\right]},$$
where the last equality  follows from \eqref{proa} under the assumption that $\{X(1),X(2),\ldots,X(N)\}$ is a centered Gaussian process.
}
\item[b)]{According to (\ref{dfa}) the variance of $F^2(s)$ takes the form
\begin{align}
Var&\left[F^2(s)\right]=\frac{1}{[N/s]^2}Var\left[\sum_{v=1}^{[N/s]}f^2(v,s)\right]=\frac{1}{[N/s]^2}\sum_{1\leq v,u\leq[N/s]}Cov\left[f^2(v,s),f^2(u,s)\right]\nonumber\\\label{p1}
&=\frac{1}{[N/s]^2}\sum_{1\leq v,u\leq[N/s]}\!\!\Bigg[E\left[f^2(v,s)f^2(u,s)\right]-E\left[f^2(v,s)\right]E\left[f^2(u,s)\right]\Bigg].
\end{align}
Using the fact that $\{X(1),X(2),\ldots,X(N)\}$ is a general centered Gaussian process, we compute 
\begin{eqnarray}
E\left[f^2(v,s)f^2(u,s)\right]=\frac{1}{s^2}E\left[\sum_{t=1}^sY_v^2(t)\sum_{w=1}^{s}Y_u^2(w)\right]=\frac{1}{s^2}\sum_{t=1}^{s}\sum_{w=1}^{s}E\left[Y_v^2(t)Y_u^2(w)\right]\nonumber\\\label{p2}
=\frac{1}{s^2}\sum_{t=1}^{s}\sum_{w=1}^{s}\Bigg[E\left[Y_v^2(t)\right]E\left[Y_u^2(w)\right]+2 \left(E\left[Y_v(t)Y_u(w)\right]\right)^2\Bigg],
\end{eqnarray}
where the last equality follows from Isserlis theorem for $4$-th joint moment of multivariate normal distribution \cite{Iss18}. By substituting \eqref{p2} and \eqref{proa} into \eqref{p1} we get the variance.
}
\end{itemize}

\end{proof}

\section{Expected value of DFA for exemplary Gaussian processes}
In this section we consider two exemplary Gaussian processes, namely Gaussian white noise (WN) and the autoregressive fractional  moving average (ARFIMA(0,d,0)) process with parameter $0<d<0.5$. The first one belongs to the family of processes with a short memory, while the second one  - to the family of long memory processes. For both processes we calculate the expected values of $F^2(s)$ analytically by using the methodology presented in section \ref{main}. We then compare the analytical expected values with the empirical ones obtained from simulated trajectories of the bot processes considered.

\subsection{Gaussian white noise}
In this part we consider the Gaussian white noise $\{Z(t)\}\sim WN(0,\sigma^2)$ which is a Gaussian process $\{Z(t)\}$ with $E[Z(t)]=0$ and the autocovariance function
\begin{equation}\label{wn_acf}
E[Z(t)Z(t+\tau)] = \begin{cases}
     \sigma^2,& \mbox{if}\; \tau=0 \\
      0, & \mbox{if}\; \tau\neq 0.
   \end{cases}
\end{equation}
Now let us denote $\{Z(1),Z(2),\ldots,Z(N)\}$ a trajectory of $WN(0,\sigma^2)$ of length $N$. According to Theorems \ref{th1} and \ref{th2} we can write
\begin{equation*}\begin{split}
E[F^2(s)] &=  \frac{1}{[N/s]} \sum_{v=1}^{[N/s]}  \sum_{i=1}^s E\left( Z^2(i+d_v) \right) \frac{1}{s}\sum_{t=1}^s\left[ \delta_{i,t} - P_1(i,t) \right]\left[ \delta_{i,t} - P_1(i,t) \right] \\
&=  \frac{1}{[N/s]} \sum_{v=1}^{[N/s]}  \sum_{i=1}^s  E\left(Z^2(i)\right)  \frac{1}{s}\sum_{t=1}^s\left[ \delta_{i,t} - P_1(i,t) \right]\left[ \delta_{i,t} - P_1(i,t) \right] \\
&=   \sum_{i=1}^s E \left( Z^2(i) \right)  \frac{1}{s}\sum_{t=1}^s\left[ \delta_{i,t} - P_1(i,t) \right]\left[ \delta_{i,t} - P_1(i,t) \right] \\
&=   \sigma^2 \sum_{i=1}^s  \frac{1}{s}\sum_{t=1}^s\left[ \delta_{i,t} - P_1(i,t) \right]\left[ \delta_{i,t} - P_1(i,t) \right],
\end{split}\end{equation*}
where $P_1(\cdot,\cdot)$ is given in (\ref{we2}).
One can show that the following holds
\begin{equation*}\begin{split}
\sum_{i=1}^s \frac{1}{s}\sum_{t=1}^s\left[ \delta_{i,t} - P_1(i,t) \right]\left[ \delta_{i,t} - P_1(i,t) \right] &= \frac{1}{s} \sum_{i=1}^s \left( 1- 2P_1(i,i) + \sum_{t=1}^s P_1^2(i,t) \right) 
=\frac{s-2}{s}.
\end{split}\end{equation*}
Thus, finally we have
\begin{equation}\label{wne}
E[F^2(s)]  = \sigma^2 \frac{s-2}{s}.
\end{equation}
In Fig. \ref{fig1} we demonstrate a comparison between the theoretical expected value of DFA for $WN(0,\sigma^2)$ given in (\ref{wne}) and the empirical one for the simulated trajectories of the considered process for $s=N/4$. We assume $\sigma=1$, length of the simulated trajectories $N=100$ and the number of Monte Carlo simulations $M=500$. As one can see, the theoretical and empirical DFAs coincide. We see that asymptotically, $E[F^2(s)]$ becomes constant, which is exactly the scaling which one expects if the considered process is the white noise.\\
\begin{figure}[ht!]
\begin{center}
\includegraphics[scale=0.6]{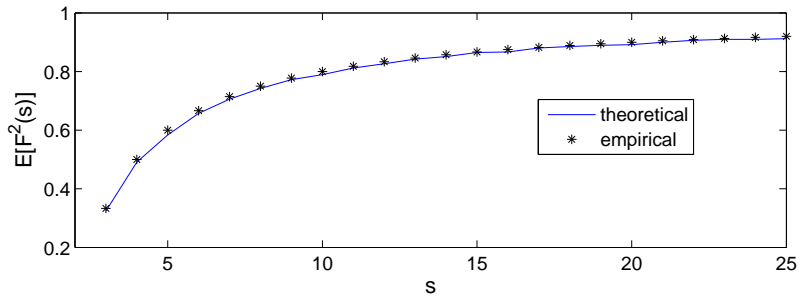}
\caption{The comparison between theoretical and empirical expected values for DFA for $WN(0,\sigma^2)$ for $\sigma=1$, $N=100$ (the trajectory length). The empirical expectation was calculated with $M=500$ trajectories of the Gaussian white noise.}
\end{center}\label{fig1}
\end{figure}

\subsection{ARFIMA$(0,d,0$) process}
As the second example we consider ARFIMA$(0,d,0)$ process $\{U(t)\}$ with $0<d<0.5$. It is a  simple example of general class of ARFIMA$(p,d,q)$ time series. The process $\{U(t)\}$ is stationary and satisfies the following equation {\cite{bib:Granger, bib:Hosking}}
\begin{equation}
\label{arfima:def}
\left(1-B\right)^dU(t) = Z(t),
\end{equation} 
for $0<d<0.5$, and $\{Z(t)\}$ constitutes a sample of independent identically distributed Gaussian random variables with zero mean and unit variance. Therefore for $d=0$ the ARFIMA$(0,d,0)$ process reduces to a Gaussian white noise $WN(0,1)$ defined in the previous subsection. 
The operator $\left(1-B\right)^d$ can be defined by the infinite binomial expansion in the form  
\begin{equation*}
\left(1-B\right)^d = \sum\limits_{k=0}^\infty \pi_jB^j,
\end{equation*}
where
$\pi_j = \frac{\Gamma(j-d)}{\Gamma(-d)\Gamma(j+1)} = \prod_{0<k\leq j}\frac{k-1-d}{k},\;j=0,1,2,\cdots$
and $B$ is the backward shift operator defined as $B^jU(t) = U(t-j)$.
The process $\{U(t)\}$ can be mapped onto an infinite moving average model, thus the following representation also holds
\begin{equation}
\label{arfima:moving}
U(t) = \sum\limits_{j=0}^{\infty} a_j Z(t-j), 
\end{equation}
where
$a_j = \frac{\Gamma(j+d)}{\Gamma(d)\Gamma(j+1)} = \prod_{0<k\leq j}\frac{k-1+d}{k},\;j=0,1,2,\cdots$
The autocovariance function of the process $\{U(t)\}$ is given by
\begin{equation}
\label{arfima:cov}
E[U(i)U(i+\tau)] = \frac{\Gamma(\tau+d)\Gamma(1-d)}{\Gamma(\tau+1-d)\Gamma(d)}.
\end{equation}
Now let us assume $\{U(1),U(2),\ldots,U(N)\}$ is a trajectory of ARFIMA$(0,d,0)$ of length $N$ for $0<d<0.5$. Following Theorems \ref{th1} and \ref{th2} one obtains
\begin{equation*}
E[F^2(s)]=  \frac{1}{[N/s]s}  \sum_{v=1}^{[N/s]} \sum_{m=1}^s E\left[\left(\sum_{i=1}^s U(i+d_v)(\delta_{im}-P_1(i,m))\right)^2\right],
\end{equation*}
where $P_1(\cdot,\cdot)$ is given in (\ref{we2}). Expanding the square in the above sum we get
\begin{equation*}
E[F^2(s)]=  \frac{1}{[N/s]s} \sum_{v=1}^{[N/s]} \sum_{m=1}^s \sum_{i,j=1}^s E[U(i+d_v)U(j+d_v)](\delta_{i,m}-P_1(i,m))(\delta_{j,m}-P_1(j,m)).
\end{equation*}
Now we will use two important facts. The first one is that the autocovariance function $E[U(i+d_v)U(j+d_v)]$ does not depend on $d_v$ because $\{U(t)\}$ is a stationary process. The second fact follows from the equality
\begin{equation*}
\sum_{m=1}^s(\delta_{i,m}-P_1(i,m))(\delta_{j,m}-P_1(j,m)) = \delta_{i,j}-2P_1(i,j) +\sum_{m=1}^s P_1(i,m)P_1(j,m).
\end{equation*}
We thus have
\begin{equation*}
E[F^2(s)]=   \frac{1}{s}  \sum_{i,j=1}^s E[U(i)U(j)]\left(\delta_{i,j}-2P_1(i,j) +\sum_{m=1}^s P_1(i,m)P_1(j,m)\right).
\end{equation*}
We order the above sum according to the time lag $\tau$ and get
\begin{equation}\begin{split}
E[F^2(s)]&=   \frac{1}{s}  \Bigg[\sum_{i=1}^s E[U^2(i)]\left(1-2P_1(i,i) +\sum_{m=1}^s P_1^2(i,m)\right)\\
&+2\sum_{\tau=1}^{s-1} \sum_{i=1}^{s-\tau} E[U(i)U(i+\tau)]\left(-2P_1(i,j) +\sum_{m=1}^s P_1(i,m)P_1(j,m)\right)\Bigg].
\end{split}\end{equation}
Now taking the autocovariance function of ARFIMA$(0,d,0)$ given in (\ref{arfima:cov}) for large values of $s$ we obtain the asymptotic behavior of $E[F^2(s)]$, 
\begin{equation}\label{arfimae}
E[F^2(s)] \approx 1- 2\frac{(3-2d)\Gamma(1-d)}{d(1+2d)(3+2d)\Gamma(d)}s^{2d-1}.
\end{equation}
We notice that for $d\rightarrow 0$ the prefactor in (\ref{arfimae}) behaves like
\begin{equation}
2\frac{(3-2d)\Gamma(1-d)}{d(1+2d)(3+2d)\Gamma(d)} \rightarrow 2,
\end{equation}
and then the result coincides with the expected value of $F^2(s)$ for $WN(0,1)$, see formula (\ref{wne}).\\
In order to demonstrate the behavior of $E[F^2(s)]$ for ARFIMA$(0,d,0)$ for $0<d<0.5$ in Fig. \ref{fig22} we plot the theoretical expected value according to formula (\ref{arfimae}) for $N=100$ and two selected values of $d$, namely $0.2$ and $0.4$. As before, we take $s=N/4$. The theoretical formulas are compared with the empirical ones obtained with $500$ Monte Carlo simulations. As one can see both functions coincide perfectly for selected values of $d$. The power law exponent is $2d-1$, as expected. For large $s$ the mean value tends to $1$ for both values of the parameter $d$.
\begin{figure}[ht!]
\begin{center}
\includegraphics[scale=0.6]{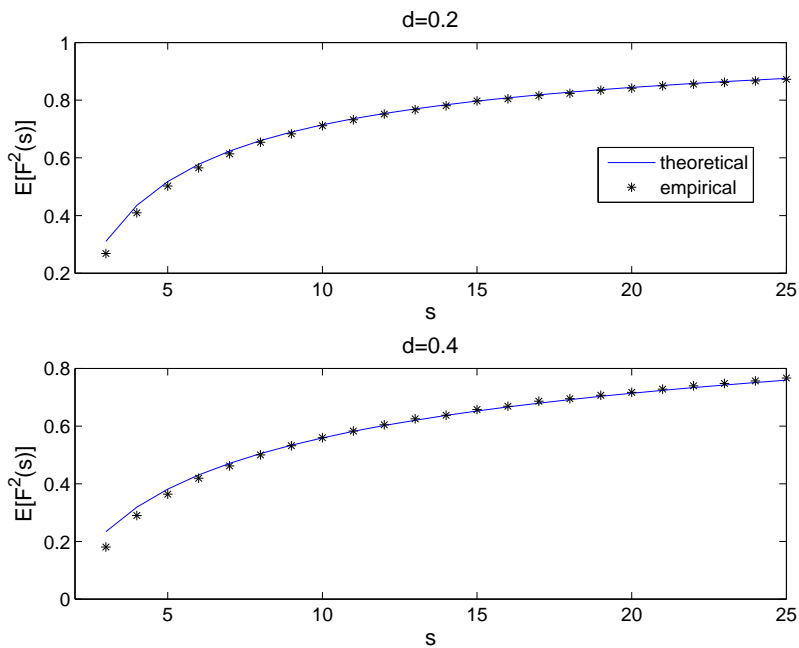}
\caption{The comparison between theoretical and empirical expected values of $F^2(s)$ for ARFIMA$(0,d,0)$ for $d=0.2$ (top panel) and $d=0.4$ (bottom panel). The trajectory length is $N=100$. The empirical expectation was calculated with $M=500$ trajectories of ARFIMA$(0,d,0)$ process.}
\end{center}\label{fig22}
\end{figure}

\section*{Acknowledgements}
AW and JG would like to acknowledge a support of National Center of Science Opus Grant No. 2016/21/B/ST1/00929 "Anomalous diffusion processes and their applications in real data modelling". AVC acknowledges the support by the Deutsche Forschungsgemeinschaft within the project ME1535/6-1.

\end{document}